\documentclass[submission,copyright,creativecommons]{eptcs}
\providecommand{\event}{GandALF 2016}
\usepackage{graphicx}
\usepackage{tikz}
\usepackage{amssymb,amsthm}
\usepackage{latexsym}

\newtheorem{theorem}{Theorem}[section]

\newtheorem{proposition}[theorem]{Proposition}
\newtheorem{lemma}[theorem]{Lemma} 
  
\newtheorem{definition}[theorem]{Definition}

\newtheorem{rema}[theorem]{Remark} 
\newtheorem{example}[theorem]{Example} 
\theoremstyle{plain}

\theoremstyle{definition}
\def\mvsquareforeod{\hbox{$\lhd$}}
\def\mveod{\ifmmode\mvsquareforqed\else{\unskip\nobreak\hfil
\penalty50\hskip1em\null\nobreak\hfil\mvsquareforeod
\parfillskip=0pt\finalhyphendemerits=0\endgraf}\fi}

\def\mvsquareforqed{\hbox{\textsc{qed}}}
\def\mvqed{\ifmmode\mvsquareforqed\else{\unskip\nobreak\hfil
\penalty50\hskip1em\null\nobreak\hfil\mvsquareforqed
\parfillskip=0pt\finalhyphendemerits=0\endgraf}\fi}

\newenvironment{tbs}{%
   \begin{itemize}\tt }{\end{itemize}}
\newcommand{\btbs}{\begin{tbs}}                                                                      
\newcommand{\etbs}{\end{tbs}}

\newcommand{\bbB}{\mathbb{B}}

\newcommand{\bbM}{\mathbb{M}}
\newcommand{\bbN}{\mathbb{N}}

\newcommand{\bbX}{\mathbb{X}}

\newcommand{\id}{\mathit{id}}

\renewcommand{\phi}{\varphi} % nicer \phi

\newcommand{\FV}{\mathit{FVar}}
\newcommand{\BV}{\mathit{BVar}}

\newcommand{\G}{\mathcal{G}}
\newcommand{\eloi}{\exists}
\newcommand{\abel}{\forall}
\newcommand{\Win}{\mathit{Win}}

\newcommand{\M}{\mathbb{M}}

\newcommand{\sse}{\subseteq}

\newcommand{\Inf}{\mathit{Inf}}

\newcommand{\Pow}{\mathcal{P}}

\newcommand{\game}{\mathcal{G}}

\newcommand{\abe}{\forall}

\newcommand{\Ord}{\textsc{Ord}}
\newcommand{\Clp}{\mathrm{Clp}}

\newcommand{\Cl}{\mathrm{Cl}}
\newcommand{\Int}{\mathrm{Int}}
\newcommand{\Prop}{\mathrm{Prop}}
\newcommand{\Lang}{\mathcal{L}}
\newcommand{\lsem}{[\![}
\newcommand{\rsem}{]\!]}
\newcommand{\lfp}{\mathrm{lfp}}
\newcommand{\gfp}{\mathrm{gfp}}
\newcommand{\operator}[3]{{#3}^{#2}_{#1}}
\newcommand{\Ev}{\mathcal{E}}
\newcommand{\unf}{\mathrm{unfold}}

\newcommand{\cthanks}{\thanks{Supported by EPSRC grant EP/N015843/1}}

\usepackage{enumerate}

\title{Games for Topological Fixpoint Logic}
\author{Nick Bezhanishvili%\nthanks
\institute{%Engineering Systems and Services\\
ILLC, University of Amsterdam\\
Amsterdam, The Netherlands}
\email{n.bezhanishvili@uva.nl}
\and
Clemens Kupke \cthanks
\institute{%Mathematically Structured Programming Group\\
University of Strathclyde\\
Glasgow, United Kingdom}
\email{clemens.kupke@strath.ac.uk}
}
\def\titlerunning{Games for Topological Fixpoint Logic}
\def\authorrunning{N.~Bezhanishvili \& C.~Kupke}

\begin{document}
 \maketitle
%\begin{frontmatter}
%   \thanks[myemail]{Email:
%    \href{mailto:N.Bezhanishvili@uva.nl} {\texttt{\normalshape
%        N.Bezhanishvili@uva.nl}}} \thanks[coemail]{Email:
%    \href{mailto:clemens.kupke@strath.ac.uk} {\texttt{\normalshape
%        clemens.kupke@strath.ac.uk}}}
\begin{abstract} 

Topological fixpoint logics are a family of logics that admits topological models and where  
the fixpoint operators are defined with respect to the topological interpretations. Here we consider
a topological fixpoint logic for relational structures based on Stone spaces,  
where the fixpoint operators are interpreted via clopen sets. 
We develop a game-theoretic semantics for this logic. First we introduce games characterising clopen  fixpoints of monotone operators on Stone spaces. These fixpoint games allow us to characterise the semantics for our topological fixpoint logic using a two-player graph game.
Adequacy of this game is the main result of our paper. Finally, we define bisimulations for the topological structures under consideration and 
use our game semantics to prove that the truth of a formula of our topological fixpoint logic is bisimulation-invariant.

  %We discuss game-theoretic characterisations of \dots and we use
  %these characterisations in order to obtain games for the topological $\mu$-calculus.
\end{abstract}
%\begin{keyword}
%	Modal mu-calculus, topology, clopen sets, regular open sets,  parity games.
%\end{keyword}
%\end{frontmatter}
\section{Introduction}\label{intro}

By \emph{topological fixpoint logics} we mean a family of fixpoint logics that admit topological models and where the fixpoint operator is defined with respect to topological 
interpretations.
In the standard semantics  fixpoint operators are interpreted as the least (or greatest) fixpoint of a monotone map in the powerset lattice. In our topological setting we interpret fixpoint operators as the least (or greatest) fixpoint of a monotone map on some (topological) sublattice of the powerset lattice
(e.g., clopen subsets, open or closed subsets, regular open or closed subsets etc.). An important motivation for studying such formalisms is that every axiomatic system
of the modal $\mu$-calculus is complete with respect to the topological semantics via clopen sets  \cite{AKM95}. Moreover, the powerful Sahlqvist completeness and correspondence result from modal logic 
can be extended to the axiomatic systems of modal $\mu$-calculus for this semantics \cite{BH12a}. 
We note that completeness results for axiomatic systems of modal $\mu$-calculus with the standard semantics are very rare, and require highly complex machinery
\cite{Kozen83},  \cite{Wal95}, see also \cite{SantVen10} and \cite{tenCF10}. 
Note also that  axiomatic systems of modal conjugated $\mu$-calculus axiomatized by Sahlqvist formulas are closed under 
Dedekind-MacNeille completions via topological semantics  \cite{Bh12b}. However, these systems  
are not closed under Dedekind-MacNeille completions for the  standard semantics  \cite{Sant08}. 
Another motivation for studying topological semantics of fixpoint logic is that it provides 
an alternative view on fixpoints operators with new notions of expressivity and definability. For a comprehensive discussion on the importance 
of generalized models in logic, including modal fixpoint logic, we refer to  \cite{AvBBN14}. 
A rather different approach to interpret fixpoint formulas over topological spaces is taken in~\cite{GoldIan} where formulas are interpreted in the full powerset lattice and where modalities are interpreted via topological operations such as closure and topological derivative. 

We illustrate the difference between standard and topological fixpoint operators with an example. Consider the frame $(\bbN \cup \{\infty\},R)$ drawn in Figure~\ref{fig:ex}. 
We assume that the topology on the set is such that clopen sets are finite subsets of $\mathbb{N}$ and cofinite sets containing the point $\infty$. 
%Then it is easy to see that t
The denotation of the formula $\Diamond^\ast p$ is the set of points that ``see points in $p$ wrt the transitive closure of the relation $R$''. Therefore
%it is easy to see that 
$\Diamond^\ast p$ is equal to 
the set $\mathbb{N}$. Indeed, $\mathbb{N}$ is the least fixed point of the map $S \mapsto \{0\} \cup \Diamond S$, where $\Diamond S = \{s' \mid \exists s \in S. (s',s) \in R \}$. 
However, if we are looking for a least clopen fixpoint of this map then we see that this will be the set $\mathbb{N}\cup \{\infty\}$. 
Intuitively, the denotation of the formula $\Diamond^\ast p$ wrt the clopen semantics is the set of all points that  ``see points in $p$ wrt the \emph{topological transitive closure} of the relation $R$''.
%\marginpar{Do we need to note that?} 
Note that a similar operation was used in \cite{Venema04} for characterising in dual terms subdirectly irreducible modal algebras. 
%This operator can be understood as 
%the set of points that are on the boundary of the transitive closure of the set $p$. 

\begin{figure}
\begin{center}
\begin{tikzpicture}[->,thick]

%\node (isolated) at ( 0,-2) [circle,fill=black,inner sep=1pt,minimum size=8pt] {};
\node (0) at ( -3,0) [circle,fill=black,inner sep=1pt,minimum size=8pt,label=above:$0$, label=below:$p$] {};
\node (1) at ( -1.8,0)  [circle,fill=black,inner sep=1pt,minimum size=8pt,label=above:$1$] {};
\node (2) at ( -0.6,0)  [circle,fill=black,inner sep=1pt,minimum size=8pt,label=above:$2$] {};
\node (infty) at ( 2.5,0) [circle,fill=black,inner sep=1pt,minimum size=8pt,label=above:$\infty$] {};
\node (d1) at ( 0,0) [circle,fill=black,inner sep=1pt,minimum size=2pt] {};
\node (d2) at ( 0.5,0) [circle,fill=black,inner sep=1pt,minimum size=2pt] {};
\node (d3) at ( 1,0) [circle,fill=black,inner sep=1pt,minimum size=2pt] {};
%\path (isolated) edge (1);
%\draw [->] (isolated) -- (1);
%\draw [->,thick] (isolated) -- (0);
%\path (isolated) edge (2);
%\draw [->,thick] (isolated) -- (infty);
\draw [->,thick] (1) -- (0);
\draw [->,thick] (2) -- (1);
\path (infty) edge [loop right] (infty);
\end{tikzpicture}\vspace{-1cm}
\end{center} 
\label{fig:ex}
\caption{Example}
\end{figure}

In this paper, we aim to advance the study of topological fixpoint logics by developing a game semantics for them. 
We will concentrate on a variant of topological fixpoint logic based on interpretations via clopen sets. 
For clopen sets we consider Stone spaces with a binary relation (\emph{descriptive $\mu$-frames} in the terminology of \cite{AKM95} and \cite{BH12a}).
The advantage of clopen sets is that the denotation of modal formulas in clopen sets is the same as in the standard Kripke semantics of modal logic. The negation of 
a formula is interpreted as the complement, conjuction and disjunction as the intersection and union, respectively, and the modal operators are also interpreted in the standard way. 
However, clopen sets of an arbitrary Stone space do not form a complete lattice and therefore the fixpoint operators, in general, may not be interpreted in Stone spaces
with the clopen semantics. Therefore, we need to restrict to a class of Stone spaces where these operators can be interpreted. We will achieve this by looking at 
relational structures based on \emph{extremally disconnected spaces} which is  a subclass of 
descriptive $\mu$-frames. 

There are several motivations for developing the game semantics for the topological $\mu$-calculus. Firstly, the semantics of a formula can be
usually much better understood when formulated in terms of games. This is especially true for formulas with some non-trivial
interplay of least and greatest fixpoint operators. Secondly, a game semantics is crucial for the development of automata-theoretic
methods of the topological $\mu$-calculus:  the game semantics provides an ``operational'' semantics for the formulas of the logic and
the definition of a run of an automata (or of its acceptance game) is entirely based on this operational view on the  truth of a formula. 
Thirdly, the game semantics is an important tool for developing the model-theory of the topological $\mu$-calculus.

The main contribution of this paper is a game semantics for the topological $\mu$-calculus based on clopen sets. 
Technically, the main result is the proof of adequacy of our game semantics. Finally we are demonstrating how the game semantics
can be used in order to obtain model-theoretic results: we prove that the topological $\mu$-calculus is invariant under what we call
clopen bisimulations.

We view the results in this paper as  first steps towards a full  theory of topological fixpoint logics. An ultimate goal is to define  
game semantics and automata for all descriptive $\mu$-frames (not necessarily based on extremally disconnected spaces). 
This would enable us to apply the methods of games and automata for tackling problems such as decidability and the finite model property 
of  axiomatic systems of the modal $\mu$-calculus. 
These systems are complete for descriptive $\mu$-frames, whereas their completeness for the standard Kripke semantics is quite problematic.

\section{Preliminaries}

\subsection{Two Player graph games}

Two-player infinite graph games, or \emph{graph games} for short, are defined
as follows.
For a more comprehensive account of these games, the reader is referred 
to~\cite{grae:auto02}.

%First some preliminaries on sequences.
%Given a set $A$, let $A^{*}$ and $A^{\om}$  denote the 
%collections of finite and infinite sequences over $A$, respectively.
%(Thus, $A^{\star} = A^{*} \cup A^{\om}$.)

A graph game is played on a \emph{board} $B$, that is, a set of \emph{positions}.
Each position $b \in B$ \emph{belongs} to one of the two \emph{players}, $\eloi$
(\'Eloise) and $\abel$ (Ab\'elard). 
Formally we write $B = B_{\eloi} \cup B_{\abel}$, and for each position $b$ we 
use $P(b)$ to denote the player $i$ such that $b \in B_{i}$.
Furthermore, the board is endowed with a binary relation $E$, so that each
position $b \in B$ comes with a set $E[b] \sse B$ of \emph{successors}. Note
that we do not require the games to be strictly alternating, i.e., successors of
positions in $B_\eloi$ or $B_\abel$ can lie again in $B_\eloi$ or $B_\abel$, respectively.
Formally, we say that the \emph{arena} of the game consists of a directed
two-sorted graph $\bbB = (B_{\eloi} ,  B_{\abel}, E)$.

A \emph{match} or \emph{play} of the game consists of the two players moving
a pebble around the board, starting from some \emph{initial position} $b_{0}$.
When the pebble arrives at a position $b \in B$, it is player $P(b)$'s turn
to move; (s)he can move the pebble to a new position of their liking, but 
the choice is restricted to a successor of $b$.
Should $E[b]$ be empty then we say that player $P(b)$ \emph{got stuck} at
the position.
A \emph{match} or \emph{play} of the game thus constitutes a (finite or 
infinite) sequence of positions $b_{0}b_{1}b_{2}\ldots\ $ such that 
$b_{i}Eb_{i+1}$ (for each $i$ such that $b_{i}$ and $b_{i+1}$ are defined).
A \emph{full play} is either (i) an infinite play or (ii) a finite play in
which the last player got stuck.
A non-full play is called a \emph{partial} play. 
%The rules of the game associate a \emph{winner} and (thus) a \emph{looser} for
Each full play of the game has a \emph{winner} and a \emph{loser}.
A finite full play is lost by the player who got stuck; the winning condition
for infinite games is usually specified using 
a so-called \emph{parity function}. In our paper, however, we specify the winning conditions
on infinite games in more intuitive terms, stating explicitly which infinite plays will be won
by which player. Throughout the paper the reader should take it for granted that the winning conditions 
involved could easily be encoded using suitable parity functions. 
%Therefore we can assume that the games we are dealing with
%are history-free determined (see below).
%
%$\Om: B \to \om$
%with finite range and $\eloi$ is the winner of $\beta \in
%B^{\om}$ if  $\max \big( \Inf(\Om\circ\beta) \big)$ is even.
%In words, $\eloi$ wins a match if the highest parity encountered infinitely
%often during the match, is even.

%In a \emph{parity game}, the set $\Ref$ is defined in terms of a \emph{parity
%function} on the board $B$, that is, a map $\Om: B \to \om$ with finite
%range.
%More specifically, the set $\Ref$ is defined by
%\begin{equation}
%\label{e-p-1}
%B^{\om}_{\Om} := \{ \beta \in B^{\om} \mid
%   \max \big( \Inf(\Om\circ\beta) \big) \mbox{ is even} \}
%\end{equation}
%(where $\Inf$ was defined at the beginning of this subsection).

A \emph{strategy} for player $i$ tells a player how to play to at a given game position: 
this can be represented as a \emph{partial} function mapping partial plays $\beta
= b_{0}\cdots b_{n}$ with $P(b_{n}) = i$ to legal next positions, that
is, to elements of $E[b_{n}]$, and that it is undefined if 
$E[b_{n}] = \emptyset$.
A strategy is \emph{history free} if it only depends on the current 
position of the match, and not on the history of the match. 
%Occasionally, it will be convenient to extend the name `strategy' to
%arbitrary functions mapping partial plays to positions; in order words, we
%allow strategies enforcing \emph{illegal} moves.
%In this context, the strategies proper, that is, the ones that dictate
%admissible moves only will be called \emph{legitimate}.
A strategy is \emph{winning for player $i$} from position $b \in B$ if it
guarantees $i$ to win any match with initial position $b$, no matter how
the adversary plays --- note that this definition also applies to positions
$b$ for which $P(b) \neq i$.
A position $b \in B$ is called a \emph{winning position} for player $i$, if 
$i$ has a winning strategy from position $b$; the set of winning positions
for $i$ in a game $\G$ is denoted as $\Win_{i}(\G)$.
Parity games enjoy \emph{history-free determinacy}, ie., at each position
of the game board one of the player has a history free winning strategy
(cf.~\cite{most:game91,emer:tree91}).

\subsection{Tarski's fixpoint game}

Recall that on any complete lattice the least fixpoint $\mu F$ and the greatest
fixpoint $\nu F$
of a monotone function $F$ exist and can be obtained as follows:
first we define for each ordinal $\alpha \in \Ord$ two sequences
$\{F^\mu_\alpha\}_{\alpha \in \Ord}$ and $\{F^\nu_\alpha\}_{\alpha \in \Ord}$ by putting
\[
\begin{array}{rclcrclcrcl}
	F^\mu_0 & = & \perp, &  \quad &  F^\mu_{\alpha+1} = F (F^\mu_{\alpha})  & \mbox{ and } &
	F^\mu_{\alpha} & = &  \bigvee_{\beta < \alpha} F^\mu_\beta  \quad  \mbox{for } \alpha \mbox{ a limit ordinal.} \\
	F^\nu_0 & = & \top, &  \quad &  F^\nu_{\alpha+1} = F (F^\nu_{\alpha})  & \mbox{ and } & 
	F^\nu_{\alpha} & = &  \bigwedge_{\beta < \alpha} F^\nu_\beta  \quad \mbox{for } \alpha \mbox{ a limit ordinal.} \\
%	 \nu f & = & \bigwedge_{\alpha \in \Ord} f^\alpha(\top) 
\end{array}
\]

The core of the game-theoretic semantics of the modal $\mu$-calculus is based on
Tarski's game-theoretic characterisation of fixpoints. Given a 
monotone function $F: \Pow X \to \Pow X$,  the game board of the standard 
fixpoint game
is defined as follows:

\begin{center}\label{Tarski_game}
\begin{tabular}{r|c|l}
Position & Player & Moves \\ %& $\Om$ \\
\hline \hline 
		$x \in X$ & $\exists$ & $\{ C \subseteq X \mid
		x \in F(C) \}$  \\
%		$(C,\forall) $ & $\forall$ & $\{ O \in \Op{X} \mid 
%			O \cap C \not= \emptyset \}$ \\
		$C \subseteq X$ & $\forall$ & $C$  
	\end{tabular} %\marginpar{types of left column in separate column?}
	\end{center}
%where $i=1$ in the \emph{least} fixpoint game and 
%$i=0$ if the \emph{greatest} fixpoint game. % is supposed to characterise the greatest fixpoint of $F$.
We will use the above notation in the following to introduce graph games: the table specifies that
$B_\exists = X$, $B_\forall = \Pow X$ and in the third column of the table the successors of each game board
position are specified. 
%Finally, the last column specifies the value of the parity function
%$\Om: B \to \om$ on a given state. 
The condition on infinite plays in the standard fixpoint game is that all infinite plays of the game are won by $\forall$ in the least fixpoint game
and by $\eloi$ in the greatest fixpoint game.
%- the reader is invited to check
%this in order to familiarise himself with the parity condition.

It is a standard result in fixpoint theory (cf. e.g.~\cite{vene:mu}) that the above least and greatest fixpoint games characterise the least and greatest fixpoint of $F$,
respectively. For example, $\eloi$ has a winning strategy at a position $x \in X$ in the least fixpoint game iff
$x$ is an element of $\mu F$. If $x$ is an element of the least fixpoint, we know that there 
exists an ordinal $\alpha$ such that $x \in F^\mu_\alpha$. In case that $\alpha$ is a limit ordinal this means
that $x \in \bigvee_{\beta < \alpha} F^\mu_\beta = \bigcup_{\beta < \alpha} F^\mu_\beta \subseteq F(\bigcup_{\beta < \alpha} F^\mu_\beta)$
where the inclusion is easily verifiable.
This means $\eloi$ can move from position $x$ to position $\bigcup_{\beta < \alpha} F^\mu_\beta$ and $\abel$ is
forced to move to some $x' \in F^\mu_\beta$ with $\beta < \alpha$. Similary, if $\alpha = \beta + 1$,
 $\eloi$ can ensure that the play reaches a position in $F^\mu_{\beta}$ after one round. In any case, due to the well-foundedness of the
 ordinals, $\eloi$ can ensure
 that the play moves from $x \in F^\mu_\alpha$ to some $x \in F^\mu_\beta$ with $\beta < \alpha$ which implies that $\eloi$ has a strategy that forces 
 $\abel$ to get stuck after a finite number of moves. 
 
 %The game characterising the greatest fixpoint of $F$ has the same game board as the one for the least fixpoint, but in this case the parity condition
 %has to guarantee that $\eloi$ wins all infinite plays.

\subsection{Topological preliminaries} 
%\begin{itemize}
%	\item Notation $\bbX=(X,\tau,R)$ where $(X,\tau)$ is a topological space.
%	\item $\Clp(\bbX)$ clopens of $(X,\tau)$ where $\bbX = (X,\tau,R)$
%	\item $\RO(\bbX)$ regular opens of $(X,\tau)$ where $\bbX = (X,\tau,R)$ 
%	\item closure $\Cl(X)$ and interior $\Int(X)$ of some set $X$
%	\item which relational structure are we considering?
%\end{itemize}

%\marginpar{consider modal spaces only as we never talk about TKF's in general?}
% A \emph{topological Kripke frame} is a triple $(X,\tau,R)$ where $\bbX=(X,\tau)$ is a topological space and $R\subseteq X^2$ a binary relation.
We will work with Kripke frames that are endowed with a topology. %carry a topological strucutre.
The most important class of such frames used in the study of modal logic is that of \emph{modal spaces} (aka \emph{descriptive frames}).
This is due to the Stone representation theorem for Boolean algebras and J\'onsson-Tarski representation theorem for
Boolean algebras with operators. 
A \emph{modal space} is a triple $(X,\tau,R)$ such that $\bbX =(X, \tau)$ is a Stone space and $R \subseteq X \times X$ is a binary relation that is \emph{point-closed}
and \emph{clopen}. The latter mean that $R(x) = \{y\in X: x R y\}$ is a closed set for each $x\in X$ and that $\Diamond U\in \Clp(\bbX)$  for each
$U \in \Clp(\bbX)$, where $\Clp(\bbX)$ is the set of all clopen subsets of $\bbX$ and $\Diamond U = \{ x \in X \mid \exists y \in U. \; x R y\}$.
Every modal algebra can be represented as the algebra $(\Clp(\bbX), \Diamond)$, where $\bbX$ is the ultrafilter space.
As a result every axiomatic system of modal logic is complete wrt modal spaces. We refer to \cite{BdRV01} for more details on
completeness of modal logics wrt modal spaces. We also note that 
modal spaces can be also represented as Vietoris coalgebras on the category of Stone spaces \cite{KKV04}.
Throughout this paper we will tacitly assume that all topological Kripke frames are modal spaces.
%\marginpar{Give modal spaces based on extremally disconnected space a name? Nick: Extremally disconnected modal spaces?}

A Stone space $\bbX=(X,\tau)$ is called \emph{extremally disconnected} if the closure of any open subset of 
$\bbX$ is open. It is well known (see e.g., \cite{Sik60}) that if $\bbX$ is an extremally disconnected space, then 
$\Clp(\bbX)$ is a complete Boolean algebra. Moreover, for a set of clopen sets $\{U_i: i\in I\}$ the infinite meets and joins are computed as:
$\bigvee \{U_i: i\in I\} = \Cl (\bigcup \{U_i: i\in I\})$ and $\bigwedge \{U_i: i\in I\} = \Int (\bigcap \{U_i: i\in I\})$.
We call a modal space $(X,\tau,R)$ an {\em extremally disconnected modal space} if $(X,\tau)$ is extremally disconnected.

%In \cite{BBH12} modal spaces where generalised to \emph{modal compact Hausdorff spaces} (MKH-spaces, for short). MKH-spaces can be viewed
%as analogues of modal spaces in the  compact Hausdorff setting. They are also represented as Vietoris coalgebras on the category of compact Hausdorff
%spaces. An MKH-space is $\bbX=(X,\tau,R)$ a point-closed topological Kripke frame such that $\Diamond_R(U)\in \Op(\bbX)$ for each
%$U \in \Op(\bbX)$, where $\Op(\bbX)$ is the set of all open sets of $(X,\tau)$, and $\Diamond_R(F)\in \Cl(\bbX)$ for each
%$F \in \Op(\bbX)$, where $\Cl(\bbX)$ is the set of all closed sets of $(X,\tau)$. A set $U\subseteq X$ is called \emph{regular open} if
%$\Int(\Cl(U)) = U$, where $\Cl$ and $\Int$ are the closure and interior operators, respectively. We let
%$\RO(\bbX)$ denote the set of all regular opens sets of $(X,\tau)$ where $\bbX = (X,\tau,R)$. Recall that regular open sets form a Boolean
%algebra where the meet  is the intersection, the negation is the interior of the complement, and the join  is the
%interior of the closure of the union.

\subsection{Modal $\mu$-calculus on topological spaces: denotational semantics}
The complete lattice structure on $\Clp(\bbX)$ of an extremally disconnected space $\bbX=(X,\tau)$ enables us to define
a topological semantics of the modal $\mu$-calculus that is based on clopen sets.
\begin{definition}
		Given a countably infinite set $\Prop$ of propositional variables ($p,q,p_0,q_1$, etc), 
		the language $\Lang_\mu$ of the modal $\mu$-calculus  is inductively defined as follows: 
		\begin{eqnarray*}
			\Lang_\mu \ni \varphi & \mathrel{::=} & p, \; p \in \Prop \mid \neg p, \; p\in \Prop \mid \varphi \wedge \varphi \mid \varphi \vee \varphi 
			\mid \bot \mid \top \mid \Diamond \varphi \mid \Box \varphi \mid \\
			& & \quad \mu p. \varphi(p,q_1,\dots,q_n) \mid 
			\nu p. \varphi(p,q_1,\dots,q_n)  
		\end{eqnarray*}
	 	where in formulas of the form $\mu p .\varphi$ and $\nu p. \varphi$ we require that the variable 
		$p$ does not occur under a negation\footnote{Formulas are always in negation normal form, ie., negations 
		only occur in front of propositional variables.}.
		The sets $\FV(\varphi)$ and $\BV(\varphi)$ of free and bound variables of a given formula
		$\varphi \in \Lang_\mu$ are defined in a standard way.
	\end{definition} 

\begin{definition}
		Given an extremally disconnected modal space $(\bbX,R)$ based on a space $\bbX=(X,\tau)$ and a valuation $V: \Prop \to \Clp(\bbX)$ 
	we define the semantics $\lsem \varphi \rsem_V^{\bbX} \in \Clp(\bbX)$ of a formula $\varphi$ 
	by induction:
		\[
		\begin{array}{rclcrcl}
			\lsem p \rsem_V & \mathrel{:=} & V(p) & \qquad & \lsem \neg p \rsem_V & \mathrel{:=} & X \setminus V(p)\\ 
			\lsem \psi_1 \wedge \psi_2 \rsem_V & \mathrel{:=} & \lsem \psi_1 \rsem_V \cap \lsem \psi_2 \rsem_V  & & 
			\lsem \psi_1 \vee \psi_2 \rsem_V & \mathrel{:=} & \lsem \psi_1 \rsem_V \cup \lsem \psi_2 \rsem_V  \\
			\lsem \perp \rsem_V & \mathrel{:=} & \emptyset & &
			\lsem \top \rsem_V & \mathrel{:=} & X \\
			\lsem \Diamond \psi\rsem_V & \mathrel{:=} & \{ x \in X \mid R(x) \cap \lsem \psi \rsem_V \not= \emptyset  \} & &
			\lsem \Box \psi \rsem_V  & \mathrel{:=} & \{ x \in X \mid R(x) \subseteq \lsem \psi \rsem_V \} \\
			\lsem \mu p .\psi \rsem_V & \mathrel{:=} & \lfp(\operator{p}{V}{\psi}) & &
			\lsem \nu p .\psi \rsem_V & \mathrel{:=} & \gfp(\operator{p}{V}{\psi}) \\
		\end{array}
		\]
		where $\operator{p}{V}{\psi}: \Clp(\bbX) \to \Clp(\bbX)$ is the (monotone) operator
		defined by $\operator{p}{V}{\psi}(U) \mathrel{:=} \lsem \psi \rsem_{V[p \mapsto U]}$ for $U \in \Clp(\bbX)$ and with
			\[ V[p \mapsto U](q)  \mathrel{:=}  \left\{ \begin{array}{l}
			U \mbox{ if } q = p \\
			V(q) \mbox{ otherwise.} 
		\end{array} \right. \] 
		We call the triple $\M = (\bbX,R,V)$ an extremally disconnected (Kripke) model 
		and write $\M[p \mapsto U]$ 
		to denote the model $\M = (\bbX,R,V[p \mapsto U])$.
		%where $p \in \Prop$ and $U \in \Clp(\bbX)$.
		\end{definition}

\section{Games for monotone operators on topological spaces}
In this section we are going to define topological analogues of the fixpoint game from page~\pageref{Tarski_game}. 
%\subsection{Extremally disconnected spaces}
We start by looking at fixpoints of a monotone function $F:\Clp(\bbX) \to \Clp(\bbX)$ 
on the lattice of clopen subsets $\Clp(\bbX)$ of an extremally disconnected Stone space $\bbX = (X,\tau)$. 
This assumption on the topology guarantees
the existence of a least and greatest fixpoint of $F$ and these fixpoints can be obtained using the ordinal approximants 
%\footnote{not sure whether this is a good name}
$F^\mu_\alpha$ and $F^\nu_\alpha$, respectively. 
To understand how the fixpoint game has to be defined we need to inspect how the ordinal approximants $F^\mu_\alpha$ and
$F^\nu_\alpha$ are computed in case $\alpha$ is a limit ordinal: 

{\small
\begin{eqnarray*}
	F^\mu_\alpha & = & \bigvee_{\beta < \alpha} F^\mu_\beta = \Cl(\bigcup_{\beta < \alpha} F^\mu_\beta) \\
	F^\nu_\alpha & = & \bigwedge_{\beta < \alpha} F^\nu_\beta = \Int(\bigcap_{\beta < \alpha} F^\nu_\beta) \\
\end{eqnarray*}}

Therefore, intuitively speaking, in order to maintain the claim that a given point $x$ is an element of $\mu F$ it suffices that $\eloi$ provides some open set
$O \subseteq X$ such that $x \in F(\Cl(O))$, so this will become easier for $\eloi$. Likewise, in order to prove that
$x \in \nu F$, $\eloi$ will now have to provide some closed set $C$ such that $x \in F(\Int(C))$ which is potentially more difficult compared to the standard fixpoint game. Note that in both cases $\Cl(O)$ and $\Int(C)$ are clopen as the closure of an open set and the interior of a closed set are clopen sets in an extremally disconnected Stone space.
Our observations form the basis for the following definitions
of the fixpoint games:

\begin{definition}
        Let $\bbX=(X,\tau)$ be an extremally disconnected topological space and let 
        $F: \Clp(\bbX)  \to \Clp(\bbX)$ be a monotone
	map. 
	%with least fixpoint $\mu f \in \Clp(\bbX)$ and largest fixpoint
	%$\nu f \in \Clp(\bbX)$. 
	We define two graph games. We start with the game board of the least fixpoint game 
	$\game^I_\mu(F)$:
%	(they only differ in the parity function)
%	$\game^I_\mu(f) = (V_\exists, V_\forall, E)$ by putting 
%	$V_\exists = X$ (``elements of $X$ are positions of $\exists$''),
%	$V_\forall = \p X$ (``elements of $\p X$ are positions of $\forall$'')
%	and 
%	by putting
	\begin{center}
	{\small
	\begin{tabular}{r|c|l}
		Position & Player & Moves \\
		\hline \hline 
		$x \in X$ & $\exists$ & $\{ C \subseteq X \mid
		x \in F(U)  \mbox{ for all } U \in \Clp(\bbX) 
		\mbox{ with } C \subseteq U \}$ \\
%		$(C,\forall) $ & $\forall$ & $\{ O \in \Op{X} \mid 
%			O \cap C \not= \emptyset \}$ \\
		$C \subseteq X$ & $\forall$ & $C$ 
	\end{tabular}}
	\end{center}
	ie, at a position $x \in X$, player $\exists$ has to move to some $C \subseteq X$
	such that $x\in F(U)$  for all clopen supersets of
	$C$ and at position $C \subseteq X$ player $\forall$ has to move to some $x' \in C$.
	Infinite plays are won by $\forall$.
	The resulting graph game will be called the least clopen fixpoint game and will be denoted by $\game^I_\mu(F)$. 
%	is complete
%	if either one of the player does not have a legal move (``gets stuck'', the play ends on a position $p$ with $E[p] = \emptyset$) or if the play is infinite. A finite complete play is lost by the 
%	player who got stuck, an infinite play is won by $\forall$.
	The greatest clopen fixpoint game $\game^I_\nu(F)$ is 
	defined similarly with the major
	difference that an infinite play is won by $\exists$. Also,
	the game board of $\game^I_\nu(F)$ reflects the aforementioned
	way of computing meets in $\Clp(\bbX)$:
	\begin{center}
	{\small
	\begin{tabular}{r|c|l}
		Position & Player & Moves \\
		\hline \hline 
		$x \in X$ & $\exists$ & $\left\{ C \subseteq X \mid
		x \in F(U)  \mbox{ for all } U \in  \Clp(\bbX) \mbox{ with } 
		\Int(C) \subseteq U \right\}$ \\
%		\upper(\{V \in \Clp(\bbX) \mid  V \subseteq C \}) \right\} $ \\
		$C \subseteq X$ & $\forall$ & $C$ 
%		$(x',\exists)$ & $\exists$ & $\{ O \in \Op{X} \mid 
%			x' \in O \}$ \\
%		$\Op{X} \ni O$ & $\forall$ & $O$
	\end{tabular}}
	\end{center}
%	where $\upper(\Gamma)$ denotes the collection of upper bounds of
%	a given $\Gamma \subseteq \Clp(\bbX)$.
\end{definition}
With these definitions at hand it is not difficult to prove that $\game^I_\mu(F)$ and $\game^{I}_\nu(F)$ indeed characterise the least
and greatest clopen fixpoints of $F$, respectively. This is the content to the following proposition.
\begin{proposition}\label{prop:first_fp_game}
	Let $\bbX= (X,\tau)$ be an extremally disconnected  space, let $F:\Clp(\bbX)\to \Clp(\bbX)$ be a monotone operator.
	 Then for any $x \in X$ we have 
	 \begin{enumerate}
	 	\item $x \in \mu F$ iff $x \in \Win_\exists(\game^I_\mu(F))$
	 	\item $x \in \nu F$ iff $x \in \Win_\exists(\game^I_\nu(F))$
	\end{enumerate}
\end{proposition}
\begin{proof}
	We only provide the proof for the greatest fixpoint game $\game^I_\nu(F)$ - the one for the least fixpoint game is very similar.
	We need to show that $\Win_\exists(\game^I_\nu(F)) = \nu F$.
	Suppose first that $x \in \nu F \in \Clp(\bbX)$. Then
	$\exists$ has an obvious winning strategy: she is playing
	the set $\nu F$. All $\forall$ can do is choosing another element $x' \in \nu F$ after which $\exists$ can move again
	to $\nu F$ and so forth. Note that any such play will be infinite and thus $\exists$ has a strategy to win any play starting at $x$, ie., $x \in \Win_\exists(\game^I_\nu(F))$.

For the converse we show that for
all ordinals $\alpha$ we have
$X \setminus F_\alpha^\nu \subseteq  \Win_\forall(\game^I_\nu(F))$ by induction on $\alpha$.	
%\begin{description}

\noindent
{\bf Case} $\alpha = 0$. Then the claim is obvious as
$X \setminus F_0^\nu = X \setminus X =\emptyset$.

\noindent
{\bf Case}  $\alpha = \beta +1$. Suppose that in a play starting
	at position $x \not\in F_\alpha^\nu = F(F_\beta^\nu)$ player $\exists$ moves to some $C \subseteq X$ with
	$x \in F(U)$ for all $U \in \Clp(\bbX)$ with  $\Int(C) \subseteq U$. 
	Clearly 
	$C \not\subseteq F_\beta^\nu$ for otherwise
	$\Int(C) \subseteq F_\beta^\nu$ and thus
	$x \in F(F_\beta^\nu) = F_\alpha^\nu$.
	Hence $\forall$ can pick
	an element $x' \in C \setminus F_\beta^\nu$.
	Now by I.H. we have that $x' \in \Win_\forall(\game^I_\nu(F))$ and thus $\forall$ has a strategy to win the play from now on. This shows that $\forall$ has a winning strategy at position
	$x$ in
	$\game_\nu(F)$ as required.
	
\noindent
{\bf Case} $\alpha$ is a limit ordinal. Consider some $x \not\in F_\alpha^\nu = \bigwedge
	F_\beta^\nu$ and let $C \subseteq X$ be chosen by $\exists$ as in the previous case. 
	By our assumption on the topology  we 		
	have $\bigwedge F_\beta^\nu = \Int (\bigcap F_\beta^\nu)$.	
	It is not difficult to see that $C \not\subseteq \bigcap F_\beta^\nu$ for suppose otherwise: then
	$\Int(C) \subseteq \Int(\bigcap F_\beta^\nu) = \bigwedge F_\beta^\nu$
	and thus $x \in F(\bigwedge F_\beta^\nu) \subseteq \bigwedge F_\beta^\nu$ which contradicts our assumption
	on $x$. Therefore there exists a $\beta < \alpha$ such that
	$C \not\subseteq F^\nu_\beta$, ie., such that there exists $x' \in C$ with $x' \not\in F^\nu_\beta$.
	By the induction hypothesis we know that $x' \in \Win_\forall(\game^I_\nu(F))$ and from position $x'$
	$\forall$ has a strategy to win the play. Therefore $\forall$ has a winning strategy
	from position $x$ as required.
%\end{description} \vspace{-5mm}
%\phantom{blabla}
%	
%	$\mathbf{\subseteq}$: Suppose that $x \not\in \nu F$. We show that then $\forall$ has a strategy
%	in $\game^I_\nu(F)$ at position $x$ such that either the game continues for one more ``round'', such
\end{proof}
This shows that the games $\game^I_\mu$ and $\game^I_\nu$ characterise the least and greatest clopen fixpoint 
of a monotone operator. %In order to define a game semantics of the topological modal $\mu$-calculus 
We will use these games
to prove adequacy of our game semantics for the topological modal $\mu$-calculus: If $\exists$ has a winning strategy in the evaluation game for a formula
of the form $\mu p. \varphi$ and $\nu p. \varphi$ then we will construct a winning strategy for her in the corresponding fixpoint games that we just discussed.
Vice versa we would like to transform winning strategies in the fixpoint games into winning strategies of the evaluation game for $\mu p.\varphi$ and
$\nu p.\varphi$. For this converse direction we will need second - but equivalent - versions of the fixpoint games. %We first discuss these alternative versions
%and then explain the relationship between the games.

\begin{definition}
%	\marginpar{consider this game in Proposition below?}
	Let $\bbX$ be an extremely disconnected space and let $F: \Clp(\bbX) \to \Clp(\bbX)$ be a monotone
	map. %with least fixpoint $\mu F \in \Clp(\bbX)$ and largest fixpoint
	%$\nu F \in \Clp(\bbX)$  
	As elements of $\Clp(\bbX)$ can occur both as position of $\exists$ and $\forall$, we clearly mark the owner of such a position using the set of markers $M = \{\exists,\forall\}$.
	We define the following two-player game $\game^{II}_\mu(F)$ by putting
	\begin{center}
	{\small
	\begin{tabular}{r|c|l}
		Position & Player & Moves \\
		\hline \hline 
		$x \in X$ & $\exists$ & $\{ (\forall,U) \in M \times \Clp(\bbX) \mid
		x \in F(U) \}$ \\
		$(\forall,U) \in M \times \Clp(\bbX)  $ & $\forall$ & $\{(\exists,U') \in M \times \Clp(\bbX) \mid U \cap U'
		\not= \emptyset \}$ \\
%		$(C,\forall) $ & $\forall$ & $\{ O \in \Op{X} \mid 
%			O \cap C \not= \emptyset \}$ \\
%		$X \ni x$ & $\exists$ & $\{ U' \in \Clp(\bbX) \mid x \in U'\}$ \\
		$(\exists,U) \in M \times \Clp(\bbX) $ & $\exists$ & U \\
	\end{tabular}}
	\end{center}
	ie, at a position $x \in X$, player $\exists$ has to move to some clopen set $U \subseteq X$
	such that $x\in F(U)$, $\forall$ challenges this by playing a element $U' \in \Clp(\bbX)$ 
	with $U \cap U' \not= \emptyset$
	%containing a clopen neighbourhood $U'$ of $x$
%	with $U \cap U' \not= \emptyset$ (one can think of this move as $\forall$
%	choosing an element $x' \in U$ and a clopen neighbourhood $U'$ of $x'$) 
	and at position $(\exists,U') \in M \times \Clp(\bbX)$ player $\exists$ has to move to some $x' \in U'$.
	Again $\forall$ wins all infinite plays of the game.
	Similarly we define the game $\game_\nu^{II}(F)$ by defining the following game board 
	and by stipulating that $\eloi$ wins all infinite plays:
	\begin{center}
		\begin{tabular}{r|c|l}
		Position & Player & Moves  \\
		\hline \hline 
		$x \in X$ & $\exists$ & $\left\{ U \in \Clp(\bbX) \mid
		x \in F(U) \right\}$ \\
%		\upper(\{V \in \Clp(\bbX) \mid  V \subseteq C \}) \right\} $ \\
		$U \in \Clp(\bbX)$ & $\forall$ & $U$  \\
%		$(x',\exists)$ & $\exists$ & $\{ O \in \Op{X} \mid 
%			x' \in O \}$ \\
%		$\Op{X} \ni O$ & $\forall$ & $O$
	\end{tabular}
	\end{center}
\end{definition}
\begin{rema}
    {\em The reader familiar with fixpoint games might be surprised and slightly worried as there is an unexpected asymmetry between the games $\game_\mu^{II}(F)$ and $\game_\nu^{II}(F)$. Both games have in fact been derived from two completely symmetric games with the following game boards (omitting the markers in $M$) and the usual winning conditions for infinite least and greatest fixpoint games:
    
   {\hspace{-3mm}\footnotesize
    \begin{center}
    	\begin{tabular}{ccc}
			\begin{tabular}{lr|c|l}
				$\game_\mu$ & Position & Pl. & Moves  \\
		\hline \hline 
				& $x \in X$ & $\exists$ & $\left\{ U \in \Clp(\bbX) \mid
		x \in F(U) \right\}$ \\
				& $U \in \Clp(\bbX)$ & $\forall$ & $U$ \\
				& $x' \in X$ & $\forall$ & $\{ U' \in \Clp{\bbX} \mid x' \in U'\}$ \\
				& $U' \in \Clp(\bbX)$ & $\exists$ & $U'$
			\end{tabular}
			& 
			\begin{tabular}{lr|c|l}
				$\game_\nu$ &	Position & Pl. & Moves  \\
		\hline \hline 
				& $x \in X$ & $\exists$ & $\left\{ U \in \Clp(\bbX) \mid
		x \in F(U) \right\}$ \\
				& $U \in \Clp(\bbX)$ & $\forall$ & $U$ \\								
				& $x' \in X$ & $\exists$ & $\{ U' \in \Clp{\bbX} \mid x' \in U'\}$ \\
				& $U' \in\Clp(\bbX)$ & $\forall$ & $U'$
			\end{tabular}
		\end{tabular}
    \end{center}
    }
    It is not difficult to see, however, that both games can be simplified to the games
    $\game_\mu^{II}(F)$ and $\game_\nu^{II}(F)$.}
\end{rema}

We will now show that games for $\mu$ and $\nu$ characterise the least and greatest clopen fixpoint. %respectively.

\begin{proposition}
	\label{prop:second_fp_game}
	Let $\bbX$ be an extremally disconnected  space, let $F:\Clp(\bbX)\to \Clp(\bbX)$ be a monotone operator.
	 Then for any $x \in X$ we have 
	 \begin{enumerate}
	 	\item $x \in \mu F$ iff $x \in \Win_\exists(\game^{II}_\mu(F))$.
	 	\item $x \in \nu F$ iff $x \in \Win_\exists(\game^{II}_\nu(F))$
	\end{enumerate}
\end{proposition}
\begin{proof}
	We first focus on the least fixpoint operator. Suppose that $x \in \mu F$ for some $x \in X$.
	Then there is a least ordinal $\alpha$ such that $x \in F^\mu_\alpha$, we  call this the $\mu$-depth of $x$. We will show that $\exists$ has a winning strategy in $\game^{II}_\mu(F)$ at $x$ by describing a strategy for $\exists$ that ensures that either $\forall$ gets stuck within the next round or that the play reaches a position $x' \in F^\mu_{\alpha'}$ with $\alpha' < \alpha$. Both facts entail that $\exists$ has a strategy such that all plays compliant with her strategy are finite and that $\forall$ is the player who will eventually get stuck.
	
%	\begin{description}
		\noindent
		{\bf Case} $\alpha = \beta + 1$. Then $x \in F^\mu_{\beta+1}=F(F^\mu_\beta)$ and $\exists$'s strategy is
		to move from $x$ to $(\forall,F^\mu_\beta)$. Player $\forall$ either gets stuck (if $F^\mu_\beta = \emptyset$) or responds by moving to some $(\exists,U')$ with  $U' \in \Clp(\bbX)$ such that $U' \cap F^\mu_\beta \not= \emptyset$. Now $\exists$ picks an arbitrary $x' \in U' \cap F^\mu_\beta$ and the round finished on a position 
		$x' \in F^\mu_\beta$ with strictly smaller $\mu$-depth as required.
		
		\noindent
		{\bf Case} $\alpha$ is a limit ordinal. Then $\exists$'s strategy is to move from $x$
		to $(\forall,\bigvee_{\beta < \alpha} F^\mu_\beta) = 
		(\forall,\Cl(\bigcup_{\beta < \alpha} F^\mu_\beta))$ which is a legal move as $x \in \bigvee_{\beta < \alpha} F^\mu_\beta \subseteq
		F(\bigvee_{\beta < \alpha} F^\mu_\beta)$. Unless $\forall$ gets stuck, he will move to some 
		position $(\exists,U')$ where $U' \in \Clp(\bbX)$ with $U' \cap \bigvee_{\beta < \alpha} F^\mu_\beta \not=
		\emptyset$. 
		In other words, the clopen subset $U'$ has a non empty intersection with the closure of
		$\bigcup_{\beta < \alpha} F^\mu_\beta$ which implies
%		Using a standard topological argument one can see that this implies
		$U' \cap \bigcup_{\beta < \alpha} F^\mu_\beta \not= \emptyset$. Therefore
		$\exists$ can pick a suitable element $x' \in \bigcup_{\beta < \alpha} F^\mu_\beta$ such that
		the round finishes in a position $x'$ of smaller $\mu$-depth. %\vspace{-0.8cm}
%	\end{description}

	We now show that
	the game $\game^\nu_{II}(F)$ characterises the greatest clopen fixpoint.
	Suppose that $x \in \nu F \in \Clp(\bbX)$. Then, as in the proof for the game
	$\game^{I}_\nu(F)$, $\exists$ has a simple winning strategy by always moving to
	$\nu F \in \Clp(\bbX)$.
	For the converse we show that for
	all ordinals $\alpha$ we have
	$X \setminus F_\alpha^\nu \subseteq  \Win_\forall(\game^{II}_\nu(F))$ by induction on $\alpha$.
	The cases $\alpha = 0$ and $\alpha = \beta +1$ follow easily from the inductive hypothesis.
	Suppose $\alpha$ is a limit ordinal and consider some
	$x \not\in F_\alpha^\nu = \bigwedge_{\beta<\alpha} F_\beta^\nu$
	and suppose that $\exists$ moves to some $U \in \Clp(\bbX)$ such that
	$x \in F(U)$. Then it is easy to see that
	$U \not\subseteq \bigcap_{\beta<\alpha} F_\beta^\nu$, for otherwise
	$U \subseteq \Int(\bigcap_{\beta<\alpha} F_\beta^\nu) = \bigwedge_{\beta<\alpha} F_\beta^\nu$
	and hence 
	$$x \in F(U) \subseteq F(\bigwedge_{\beta<\alpha} F_\beta^\nu) \subseteq \bigwedge_{\beta<\alpha} F_\beta^\nu.$$
	Therefore $\forall$ can pick some $x' \not\in \bigcap_{\beta<\alpha} F_\beta^\nu$, ie.,
	$x' \not\in F_\beta^\nu$ for some $\beta < \alpha$. By I.H. we know that
	$\forall$ has a winning strategy from position $x'$ and hence - as $\exists$'s move
	to $U$ was arbitrary - we showed that $\forall$ has a winning strategy from position $x$.
	This finishes the proof of $X \setminus F_\alpha^\nu \subseteq  \Win_\forall(\game^{II}_\nu(F))$
	which is equivalent to $F_\alpha^\nu \subseteq X \setminus \Win_\forall(\game^{II}_\nu(F)) = \Win_\exists(\game^{II}_\nu(F))$ for all $\alpha \in \Ord$. The latter implies
	$\Win_\exists(\game^{II}_\nu(F)) \subseteq \nu F$.
	%using the first half of the proposition we can conclude that 
	%$\Win_\exists(\game^{II}_\nu(F)) = \nu F$. \marginpar{possibly unclear that $\Win_\exists(\game^{II}_\nu(F)) \subseteq X$ - need to clarify}
	\end{proof}
	We conclude our discussion of fixpoint games on extremally disconnected spaces. The reader might wonder why we introduced two
	games $\game^I_\mu(F)$, $\game^{II}_\mu(F)$ for the least fixpoint of $F$ and two games for
	the greatest fixpoint. Do we really need both variants of the $\mu$- and $\nu$-games?
	The reason why %we strongly believe that 
	both variants seem necessary for proving our adequacy theorem is based on the
	following observation\footnote{We state this observation for $\mu$, but it equally applies to $\nu$.}: 
	The games $\game^I_\mu$ and $\game^{II}_\mu$ characterise both the same least fixpoints and have therefore
	the same winning regions within the set of states $X$. It is, however, in general not possible to transform   strategies
	of $\exists$ in the first variant of the $\mu$-game into corresponding strategies for $\exists$ in the second game. To see this, suppose that $\exists$ has a strategy $f$ in $\G^I = \G^I_\mu(F)$ at position
	$x$ and suppose $f(x) = C$. We would like to equip $\exists$ with a corresponding strategy $g$ in $\G^{II}=\G^{II}_\mu(F)$ at position $x$ such that
	for the next ``round'' $x U U' y$ of $\G^{II}$ that is compliant with $g$, there is a corresponding
	round  $x C y$  of $\G^{I}$  compliant with $f$ (and by re-using that argument round-by-round, one could ensure that
	$f$ is a winning strategy for $\exists$ in $\G^I$ iff $g$ is
	a winning stratgey for $\exists$ in $\G^{II}$).
	
%	all resulting $\G^{II}$-plays conform $g$ correspond to $f$-conform $\G_I$-plays that are of 
%	equal length. %and that pass through the same positions in $X$. 
%	At position $x$ it is $\exists$'s strategy is to move to some $C \subseteq X$ such that
%	$x \in F(U)$ for all clopen sets $U$ with $C \subseteq U$ and $\forall$ picks then
%	an element $x' \in C$. 
	To achieve this, we have to define $\exists$'s strategy $g$ such that she moves 
	from $x$ in $\G^{II}$ to some suitable clopen set $U$.
	Suppose $U \subseteq \Cl(C)$. Then $\forall$ can respond with some $U' \in \Clp(\bbX)$ such that
	$U \cap U' \not= \emptyset$. This implies $U' \cap \Cl(C) \not= \emptyset$ and thus - as $U'$ is  clopen - 
	that $U' \cap C \not= \emptyset$. Hence, $\exists$ can continue the play by picking an element 
	$y$ of $U' \cap C$ which overall results in the partial $\G^{II}$-play $x U U' y$. Clearly, the sequence
	$x C y$ is also an $f$-compliant $\G^I$-play and therefore can act as the
	corresponding play for the $\G^{II}$-play $x U U' y$. Similarly one can show that in any play where
	$\exists$ moves from position $x$ to some $U$ with $U \not\subseteq \Cl(C)$, $\forall$ can ensure that
	the next state $y$ that is reached in the play will be an element of $X \setminus C$ and therefore 
	that the resulting $\G^{II}$-play is no longer linked to any corresponding $f$-compliant $G^I$-play.
	
	Therefore we can construct a corresponding strategy for $\exists$ in $\G^{II}$ iff there is a legitimate
	move $U$ for $\exists$ at $x$ with $U \subseteq \Cl(C)$.
	%There is, however, a problem
	%with the construction of strategy $g$ 
	%for $\exists$ in $\G^{II}$ that we outlined in the previous paragraph: 
	In
	general, however, there is no 
	suitable clopen set $U \subseteq \Cl(C)$ with
	$x \in F(U)$ -  and this property is required for a legitimate move in $\G^{II}$.
	This is demonstrated by the following example.
	
	\begin{example} {\em 
		Consider the Stone-\u{C}ech compactification $\beta(\bbN)$  of the natural numbers\footnote{Which is extremally disconnected, see eg~\cite{Sik60}.},
		%any space of the form $\beta X$ is 
		%extremally disconnected.}, 
		let 
		$C \subseteq \beta(\bbN)$
		be the collection of non-principal ultrafilters over $\bbN$ and consider the (trivially monotone) operator 
		$$F =\id_{\Clp(\beta(\bbN))}: \Clp(\beta(\bbN)) \to \Clp(\beta(\bbN)) .$$
	 	For any clopen $U \in \Clp(\beta (\bbN))$ we have $U = \hat{S} = \{ u \in \beta(\bbN) \mid S \in u\}$
		for some suitable set $S \subseteq \bbN$. With this in mind, it is easy to see that 
		for all clopens $U$ we have $U \subseteq C$ implies $U = \emptyset$. 
		%$C \subset \hat{S}$ 
		%iff $S$ is a co-finite subset of $\bbN$ and $\Cl(C) = C$. 
		
		 Consider now an arbitrary $x \in C$. 
		We have that 
		%$U = \hat{S}$ for some
		%co-finite $S \subseteq \bbN$ and, consequently, $
		$x \in F(U)$
		for all $U \in \Clp(\beta\bbN)$ such that $C \subseteq U$  (in particular,
		$C$ would be a legitimate move in $\G^I_\mu(F)$ at $x$). 
		On the other hand, for $U \in \Clp(\beta\bbN)$ 		we have that
		$U \subseteq \Cl(C) = C$ implies $U = \emptyset$ and thus $x \not\in F(U)$ for all these $U$ (which shows that there is no suitable move for $\exists$ in $\G^{II}_\mu(F)$ at $x$ that correponds to her move from
		$x$ to $C$). }
	\end{example}
%	To sum it all up: while versions $I$ and $II$ of the fixpoint games both characterise the 
%	same least and greatest fixpoints, it is in general not possible to directly  transform winning strategies in the game in version $I$ into winning strategies in the game in version $II$.
	
%	Similarly, it seems not feasible, to 
%	establish a direct correspondence between
%	winning strategies of $\exists$ in the evaluation game and winning strategies of only one of the games $\game^I_\mu$ or $\game^{II}_\mu$.
%	\marginpar{discussion of game variants with "no substantial argument" - at least provide detailed proof of Prop. 3.2 and 3.4}
%\subsection{Compact Hausdorff spaces}
%Similarly to what we did for extremally disconnected spaces, we also will define games that characterise regular open fixpoints
%for monotone functions on the complete lattes $\RO(\bbX)$ of regular open subsets of a given compact Hausdorff space $\bbX$.
%The compactness of the space does not really play a role in the game - we are working with compact Hausdorff spaces only because
%this will allow us to obtain a reasonable interpretation of the Boolean operators of the $\mu$-calculus.
%
%\begin{definition}
%
%
%\end{definition}
%
%The alternative versions of the game look as follows:
%
%\begin{definition}
%
%\end{definition}

\section{Game semantics for the $\mu$-calculus on topological spaces}

We are now ready to define the game characterisation of the clopen semantics of the modal $\mu$-calculus.
Our presentation follows the presentation of the standard game semantics of the modal $\mu$-calculus that can be found e.g.~in~\cite{vene:mu}.
In the following we assume that we are dealing with {\em ``clean''} formulas in $\Lang_\mu$:
		\begin{definition}
			A formula $\varphi \in \Lang_\mu$ is called clean if no two distinct occurrences of fixpoint operators in $\varphi$
		bind the same propositional variable and if a variable occurs either free or bound in $\varphi$ (but not both bound and free).
		For any bound variable $p \in \Prop$ that occurs within a clean formula  $\varphi$  we denote by
		$\varphi@p = \eta p. \psi$ the unique subformula of $\varphi$  where $p$ is bound by the fixpoint operator $\eta \in \{ \mu ,\nu \}$.
		\end{definition}
The restriction to clean formulas is standard practice in the modal literature. It will simplify the game definition. Furthermore it allows us
		to give a concise definition of when the unfolding of one fixpoint variable depends on the unfolding of another one.		
		\begin{definition}
			For a clean formula $\varphi \in \Lang_\mu$ and bound variables $x,y \in \Prop$ occurring in $\varphi$ we say
			$x \leq_\varphi y$ if $\varphi@x$ is a subformula of $\varphi@y$.
		\end{definition}
%\subsection{Extremally disconnected spaces: clopen semantics}

		\begin{definition}
			Let $\varphi \in \Lang_\mu$ be a formula and let $\M=(\bbX,R,V)$ be an extremally disconnected Kripke model 
			%based 	
			%on a space $\bbX=(X,\tau)$ 
			together with valuation
			$V: \Prop \to \Clp(\bbX)$. The game board of
			the evaluation game $\Ev (\varphi,\M)$ is specified in the table in Figure~\ref{tab:evgame}.
			
			As usually a finite full play of $\Ev (\varphi,\M)$ is lost by the player who got stuck at the end of the play.
			In order to specify the winning condition on infinite plays $\pi$ we need the following notation:
			\[ \Inf(\pi) \mathrel{:=} \{p \in \BV(\varphi) \mid p \mbox{ occurs infinitely often in } \pi \} .\]
			A standard argument shows that for any infinite play $\pi$ of  $\Ev (\varphi,\M)$ the set
			$\Inf(\pi)$ is nonempty, finite and upwards directed with respect to the dependency order $\leq_\varphi$.
			Therefore the maximal element $\max(\Inf(\pi))$ of $\Inf(\pi)$ wrt $\leq_\varphi$ is well-defined
			and we declare $\exists$ to be the winner of an infinite play $\pi$ of $\Ev(\varphi,\M)$ iff
			$\max(\Inf(\pi))$ is a $\nu$-variable, ie., a variable bound by a greatest fixpoint operator.  
		\end{definition}
		After our discussion of fixpoint games, the reader should have little problems with understanding the intuition behind the winning condition:
		an infinite play during which the highest infinitely often ``unfolded'' fixpoint variable is a $\nu$-variable
		corresponds to an infinite play of a greatest fixpoint game. 
		Therefore $\eloi$ wins such a play.  Similarly all infinite plays in which the highest infinitely
		often unfolded variable is a $\mu$-variable are won by $\forall$.		
		\begin{figure} {\small
	\begin{center} 	\begin{tabular}{r|c|l}
		Position & Player & Possible Moves \\
		\hline \hline 
		   $(p,x)$, $p \in \FV(\varphi)$ and $x \not\in V(p)$ & $\exists$ & $\emptyset$ \\
		   $(p,x)$, $p \in \FV(\varphi)$ and $x \in V(p)$ & $\forall$ & $\emptyset$ \\
		   $(\neg p,x)$, $p \in \FV(\varphi)$ and $x \not\in V(p)$ & $\forall$ & $\emptyset$ \\
		   $(\neg p,x)$, $p \in \FV(\varphi)$ and $x \in V(p)$ & $\exists$ & $\emptyset$ \\
		   $(\psi_1 \wedge \psi_2,x)$ & $\forall$ & $\{ (\psi_1,x) , (\psi_2,x) \}$ \\
		    $(\psi_1 \vee \psi_2,x)$ & $\exists$ & $\{ (\psi_1,x) , (\psi_2,x) \}$ \\
		     $(\Diamond \psi,x)$ & $\exists$ & $\{ (\psi,x') \mid Rxx' \}$ \\
		       $(\Box \psi,x)$ & $\forall$ & $\{ (\psi,x') \mid Rxx' \}$ \\
		       $(\eta p. \psi,x)$, $\eta \in \{\mu,\nu\}$ & $\exists/\forall$ & $(\psi,x)$ \\
		    $(p,x)$, $p \in \BV(\varphi)$, $\varphi@p = \mu p.\psi$ & $\forall$ & $\{ (p,U) \mid U \in \Clp(\bbX), \; x \in U \}$\\
		     $(p,x)$, $p \in \BV(\varphi)$, $\varphi@p = \nu p.\psi$ & $\exists$ & $\{ (p,U) \mid U \in \Clp(\bbX), \; x \in U \}$\\
		    $(p,U)$, $p \in \BV(\varphi)$, $\varphi@p = \mu p.\psi$& $\exists$ & $\{ (\psi,x') \mid x' \in U\}$ \\
		    $(p,U)$, $p \in \BV(\varphi)$, $\varphi@p = \nu p.\psi$& $\forall$ & $\{ (\psi,x') \mid x' \in U\}$ \\
	\end{tabular}\\ \vspace{2mm}
	where $x,x'$ denote elements of $X$ and $U$ denotes a clopen subset of $\bbX=(X,\tau)$.
		\end{center}
		\caption{Game board of the evaluation game $\Ev(\varphi,\M)$}
		\label{tab:evgame}}
		\end{figure}
	We now turn to the formulation and proof of the main theorem of this section.
	First we need to introduce some terminology and an auxiliary lemma.
	\begin{definition}
		Consider a two-player graph game $\G$ with set of positions $B$.
		For a set
$Y \subseteq B$ we say a $\G$-play $\pi$ is $Y$-full if either $\pi$ is a full play or
$\pi = b_0 \dots b_n$ is a partial play with $b_0, \dots, b_{n-1} \not\in Y$ and $b_n \in Y$, i.e.,
$b_n$ is the first position of the play occurring in $Y$.
	\end{definition}
	\begin{lemma}\label{lem:GvsGU}
		Let $\M=(\bbX,R,V)$ be an extremally disconnected model, let $\varphi = \eta p.\delta$ with $\eta \in \{\mu,\nu\}$ be a fixpoint formula and consider the games 
		$\G_\eta = \Ev(\eta p.\delta,\M)$ and $\G_U=\Ev(\delta,\M[p\mapsto U])$ with $U\in \Clp(\bbX)$.
		Furthermore we let $\unf_p = \{ (p,x') \mid x' \in X \}$.
		\begin{enumerate}[(i)]
			\item Any strategy $f_\eta$ for $\exists$ in $\G_\eta$ at $(\delta,x)$ 
			corresponds to a strategy $f_U$ for $\exists$ in $\G_U$ at $(\delta,x)$ 
			such that any $\unf_p$-full, $f_\eta$-compliant $\G_\eta$-play starting at $(\delta,x)$ %$\pi = (\delta,x) b_1 \dots b_j \dots$
			%the sequence $\pi' = (\delta,x) b_1 \dots b_j \dots$ (i.e., $\pi'$ is equal to $\pi$ minus the first position) 
			is an $f_U$-compliant, full $\G_U$-play.
%			$\pi' = (\delta,x) b_1 \dots b_j \dots$ (i.e., $\pi'$ is equal to $\pi$ minus the first position).
%			\item Any strategy $f_\eta$ for $\exists$ in $\G_\eta$ at position
%			$(\eta p.\delta,x)$ can be turned into a strategy $f_U$ for $\exists$
%			in $\G_U$ at position $(\delta,x)$
%			such that any $f_\eta$-conform $\G_\eta$-play $\pi = (\eta p.\delta,x)(\delta,x) b_1 \dots b_j \dots$ 
%			that does not contain a second position of the form
%			$(\delta,x')$ with $x' \in X$, contains the $f_U$-conform $\G_U$-play
%			$\pi' = (\delta,x) b_1 \dots b_j \dots$ 
%			(i.e., $\pi'$ is equal to $\pi$ minus the first position).
			\item Any strategy $f_U$ of $\exists$ in $\G_U$ at $(\delta,x)$
			corresponds to a strategy $f_\eta$ for $\exists$ in $\G_\eta$
			at $(\delta,x)$ such that for any full $f_U$-compliant $\G_U$-play starting at $(\delta,x)$
%			$\pi=(\delta,x) b_1 \dots b_i \dots$ 
%			the sequence $\pi' = (\eta p.\delta,x) \pi$ 
			is an $f_\eta$-compliant, $\unf_p$-full $\G_\eta$-play.
		%	Moreover, if $\pi$ does not end at a position of the form $(p,x')$ with $x' \in X$
		%	the $\pi$ is a full $\G_\eta$-play if $\pi'$ is a full $\G_U$-play.
		\end{enumerate}
	\end{lemma}
	\begin{proof}
		The lemma follows from the fact that a sequence of the form 
		$\pi = (\delta,x) b_1 \dots b_j \dots$
		is an $\unf_p$-full $\G_\eta$-play iff it 
		%``tail'', i.e., the sequence 
	%	$\pi' = (\delta,x) b_1 \dots b_j \dots$,
		is a full $\G_U$-play.
%		The first observation to be made is that the games  
%		 $\game = \Ev(\varphi,\M)$ and $\game_U = \Ev(\varphi,\M[p \mapsto U])$ are very similar. More precisely, these games
%		 only differ in positions of the form $(p,x)$ at which in $\game$ the fixpoint gets unfolded (as $p$ is treated as a bound fixpoint variable)
%		 whereas in $\game_U$ a play ends at position $(p,x)$ and the winner is declared depending on whether or not
%		 $x \in U$.
%		 
%		To be provided.
	\end{proof}
	\begin{theorem}[Adequacy]
		Let $\M=(\bbX,R,V)$ be an extremally disconnected model with
		valuation $V: \Prop \to \Clp(\bbX)$. 
		For every formula $\varphi \in \Lang_\mu$ and every $x \in X$ the following are equivalent:
		\begin{enumerate}[(i)]
			\item $x \in \lsem \varphi \rsem_V$, and
			\item $\exists$ has a winning strategy at position $(\varphi,x)$ in
				$\Ev(\varphi,\M)$.
		\end{enumerate}
	\end{theorem}
	\begin{proof}[(Sketch)]
		The proof goes by induction on $\varphi$. We only will sketch the induction
		step for the case that $\varphi = \mu p .\delta$ - the full proof of the theorem
		is quite lengthy and most of the details are similar to the adequacy proof of the standard game
		semantics  for
		the modal $\mu$-calculus. We put $\G = \Ev(\varphi,\M)$ and for any clopen
		subset  $U \in \Clp(\bbX)$ we put $\game_U = \Ev(\delta,\M[p \mapsto U])$.	
		
		By the induction hypothesis on $\delta$ and because $\lsem \delta \rsem_{V[p \mapsto U]}
		=\operator{p}{V}{\delta}(U)$ we have for all $U \in \Clp(\bbX)$ that
		\begin{equation}\label{equ:IH_adequ}
		 x \in \operator{p}{V}{\delta}(U) \mbox{ iff } (\delta,x) \in \Win_\exists(\G_U) .
		 \end{equation}
%		 The first observation to be made is that the games  
%		 $\game = \Ev(\varphi,\M)$ and $\game_U = \Ev(\varphi,\M[p \mapsto U])$ are very similar. More precisely, these games
%		 only differ in positions of the form $(p,x)$ at which in $\game$ the fixpoint gets unfolded (as $p$ is treated as a bound fixpoint variable)
%		 whereas in $\game_U$ a play ends at position $(p,x)$ and the winner is declared depending on whether or not
%		 $x \in U$. \marginpar{Referee: discussion of $\game$ vs $\game_U$ too informal - add lemmas!}
		 
		 In order to prove the theorem for $\varphi = \mu p. \delta$ it suffices to show that
		 the following are equivalent:
		 \begin{eqnarray}
		 	x   &\in&  \Win_\exists \left(\game^{I}_\mu(\operator{p}{V}{\delta})\right)   \label{winone} \\
			x  & \in&  \Win_\exists \left(\game^{II}_\mu(\operator{p}{V}{\delta})\right)  \label{wintwo} \\
			(\varphi,x)   &\in&  \Win_\exists (\game).	\label{win}
		 \end{eqnarray}
		 We proved the equivalence of (\ref{winone}) and (\ref{wintwo}) in the previous section.
		 To prove all of the equivalences,  we will now 
		 show that (\ref{wintwo}) implies (\ref{win}) which in turn implies (\ref{winone}). 
		 For the implication from (\ref{wintwo}) to (\ref{win})
		 consider some state $x \in \Win_\exists (\game^{II}_\mu(\operator{p}{V}{\delta}))$, ie., $\exists$ has
		 a history-free winning strategy at position $x$ in $\game^{II}_\mu(\operator{p}{V}{\delta}))$ 
		 represented by two (possibly partial)
		 functions 
		 $$U: X \to \Clp(\bbX) \qquad \mbox{ and } \qquad N: \Clp(\bbX) \to X.$$ 
		 W.l.o.g.~we can
		 assume that $\langle U,N \rangle$ is winning for $\exists$ from all positions in $\Win_\exists (\game^{II}_\mu(\operator{p}{V}{\delta}))$ (in particular, $U$ and $N$ are defined at those positions).
		 As the strategy $U$ is winning (and thus legitimate) at all 
		 $x \in \Win_\exists (\game^{II}_\mu(\operator{p}{V}{\delta}))$ 
		 we have that 
		 for all such $x$ that $U(x)$ is a legitimate move at $x$.
		 Hence $x \in \operator{p}{V}{\delta}(U(x))$ and thus, by~(\ref{equ:IH_adequ}),
		 $(\delta,x) \in \Win_\exists(\G_{U(x)})$.		 
		 Therefore,
		 for each $x \in \Win_\exists (\game^{II}_\mu(\operator{p}{V}{\delta}))$,
		 we can assume 
		 \begin{enumerate}[(a)]
		  \item that there is a winning strategy $f_{U(x)}$ for $\exists$ in the game $\game_{U(x)}$ at position $(\delta,x)$ and
		  \item that $(\forall,U(x)) \in \Win_\exists (\game^{II}_\mu(\operator{p}{V}{\delta}))$.
		 \end{enumerate}
%		 We now argue that this means that there exists a winning strategy $f_s$ for $\exists$ in the game
%		 $\game_{\cl{T(s)}} = \Ev(\delta,\bbS[x \mapsto \cl{T(s)}])$ (this is a slight abuse of notation, because in general
%		 $\cl{T(s)}$ will only be a closed and not a clopen set, but it's still clear how that game is defined: in case we reach a position
%		 of the form $(x,t)$ we check whether or not $t \in \cl{T(s)}$ in order to see who won the play). Suppose for a contradiction
%		 that such a strategy $f_s$ of $\exists$ does not exist. Then $\forall$ has a winning strategy in $\game_{\cl{T(s)}}$
%		 at position $s$. 

		 %Because the games $\G$ and $\G_{U(x)}$ are so similar, 
		 As seen in Lemma~\ref{lem:GvsGU}, the winning strategy $f_{U(x)}$ can be (trivially) turned
		 into a valid strategy $f_{\mu,x}$ for $\exists$  in $\G$ at $(\delta,x)$ that can be followed 
		 until another position of the form $(p,x')$ is reached or until $\exists$ wins the game. 
		 This observation is important for defining
		 $\exists$'s strategy in $\game$ starting from position $(\varphi,x)$:
		 \begin{itemize}
		 	\item starting from $(\varphi,x)$, the play proceeds to $(\delta,x)$ and after that
			$\exists$ plays strategy $f_{\mu,x}$.
		 	%\item after the initial move from $(\varphi,x)$ to $(\delta,x)$ (there is no choice here) $\exists$
			%plays her strategy $f_x^{U(x)}$ 
			\item if the $f_{\mu,x}$-compliant play never reaches a position of the form 
				$(p,x')$ then $\exists$ continues playing according to $f_{\mu,x}$ and
				wins: the resulting $f_{\mu,x}$-compliant, full $\G$-play
				contains a $f_{U(x)}$-compliant full $\G_U$-play (by Lemma~\ref{lem:GvsGU}) starting at $(\delta,x)$
				which is won by $\exists$
				as $f_{U(x)}$ is a winning strategy for $\exists$ in $\G_{U(x)}$ at $(\delta,x)$.
			\item Suppose an $f_{\mu,x}$-compliant play reaches a position of the form 
				$(p,x')$. Until now - by Lemma~\ref{lem:GvsGU} - the play corresponds to
				a $f_{U(x)}$-compliant play of $\G_{U(x)}$. As $f_{U(x)}$ is a winning strategy for 
				$\exists$ in $\G_{U(x)}$ this entails that
				$x' \in U(x)$.  It is now  $\forall$'s turn
				to move in $\G$ to a position $(p,U')$ with $x' \in U'$.
				
				As $x' \in U(x) \cap U'$ (by the definition of $\G$), we have $U(x) \cap U' 
				\not= \emptyset$, i.e., the move to $(\exists,U')$
				is a legal move for $\forall$ in $\game^{II}_\mu(\operator{p}{V}{\delta})$
				at position $(\forall,U(x))$. As the latter is an element of
				$\Win_\exists (\game^{II}_\mu(\operator{p}{V}{\delta}))$, we also
				have that $(\exists, U') \in \Win_\exists (\game^{II}_\mu(\operator{p}{V}{\delta}))$.
				Hence $\exists$'s winning strategy $N$ in $\Win_\exists (\game^{II}_\mu(\operator{p}{V}{\delta}))$ specifies a well-defined, legitimate move at $U'$ that follows $\exists$'s winning strategy in $\Win_\exists (\game^{II}_\mu(\operator{p}{V}{\delta}))$.
				
				Therefore, in $\G$, $\exists$ answers $\forall$'s move to $(p,U')$ by moving to 
				$(\delta,y)$ with $y = N(U')$ and continues from there according to
				strategy $f_{\mu,y}$.
				%following her strategy
				%$N$ from the fixpoint game. 
		 \end{itemize}
		 It is not difficult to check, that
		 this describes indeed a winning strategy for $\exists$ in $\game$ from
		 position $(\varphi,x)$. The key observation is that for
		 any $\G$-play of the form \vspace{-.2cm}
		 $$ \pi = (\varphi,x) \dots (p,x_1)(\delta,U_1)(\delta,y_1) \dots (p,x_2)(\delta,U_2)(\delta,y_2)
		 \dots (p,x_i)(\delta,U_i)(\delta,y_i) \dots \vspace{-.2cm}$$		 
		 there is a corresponding infinite
		 play of $\game^{II}_\mu(\operator{p}{V}{\delta})$ of the form \vspace{-.2cm}
		 $$ \pi' = x \; (\forall,U(x))\; (\exists,U_1) \; y_1 \; (\forall,U(y_1)) \; (\exists,U_2) \; y_2 \dots (\forall,U(y_{i-1})) \; (\exists,U_i) \; y_i \dots  \vspace{-.2cm}$$
		 which is compliant with $\exists$'s winning strategy in $\game^{II}_\mu(\operator{p}{V}{\delta})$ 		 	 and where the number of fixpoint unfoldings in $\pi'$ is equal to the number of occurrences
		 of positions of the form $(p,x')$ in $\pi$. As $\pi'$ is won by $\exists$, the play $\pi'$ must end after finitely many moves. Hence there are only finitely many occurrences of positions of the form $(p,x')$ in $\pi$, i.e., from a certain position $(\delta,x')$ on the play
		 follows $\exists$'s strategy $f_{\mu,x'}$ in $\G$ at $(\delta,x')$. In other words, 
		 such a play is won by $\exists$ as - modulo a finite prefix - it
		 corresponds by our construction to a $f_{U(x')}$-compliant $\G_{U(x')}$-play from position $(\delta,x')$ and
		 $f_{U(x')}$ is a winning strategy for $\exists$ at $(\delta,x')$.	\\	 
%		  there won't be any infinite play
%		 of the game $\game$ according to $\eloi$'s strategy that passes infinitely often
%		 through a position of the form $(p,x')$ because otherwise there would be a corresponding
%		 play of the fixpoint game that is played according to $\eloi$'s winning strategy but
%		 nevertheless won by $\forall$. Therefore any play of the game $\game$
%		 according to $\eloi$'s strategy will eventually be 
%		 a play of the game $\game_{U(x)}$ that is played according to $\eloi$'s winning strategy
%		 $f_x^{U(x)}$ in $\game_{U(x)}$ and thus all such plays will be won by $\eloi$. 
	 
		 \noindent We now turn to the proof of the implication from (\ref{win}) to (\ref{winone}). 
		 Consider a strategy $f$ for $\exists$ in
		 $\game$ such that $f$ is winning for all positions in $\Win_\exists(\game)$ and let $\Delta := \{ x \in X \mid (\delta,x) \in \Win_\exists(\game)\}$. To prove our claim it suffices to show that
		 $\Delta \subseteq \Win_\exists(\game_\mu^I(\operator{p}{V}{\delta}))$ by equipping $\exists$
		 with a suitable strategy in $\G_\mu^I(\operator{p}{V}{\delta})$ that is winning at all positions
		 in $\Delta$.
%		  at position $(\varphi,x_0)$. We are going to describe a winning strategy
%		 for $\exists$ in  $\game_\mu^I(\operator{p}{V}{\delta})$. 
		 As before, we let 
		 %$\Delta := \{ x \in X \mid (\delta,x) \in \Win_\exists(\game)\}$, 
		 $\unf_p=\{(p,x)\mid
		 x \in X\}$ and
		 for all $x \in \Delta$ we put \vspace{-.25cm}
		 \begin{eqnarray*} 
		 C(x) & := & \{ z_y \in X \mid \exists y \in X. 
		 (p,y) \mbox{ is reachable in an $\unf_p$-full $\G$-play $\pi$ from } (\delta,x) \mbox{ such that} \\
		 & & \qquad \qquad \mbox{$\pi$ is compliant with
		 $\exists$'s strategy } f, \\
		 	%& & \qquad \qquad \mbox{ no proper prefix of $\pi$ contains a position of the form } (p,y'), \\ 
		 	& & \qquad \qquad \mbox{ $\forall$ can move from $(p,y)$ to position } (p,U_y)  \\
			& & \qquad \qquad \mbox{ to which $\eloi$'s reply according to her strategy $f$ is %} \\
%			& & \qquad \qquad \mbox{ 
			to move to } (\delta,z_y) \mbox{ with } z_y \in U_y \}
					 \end{eqnarray*}  
		 Let $x \in \Delta$ and let $U \subseteq X$ be clopen with $C(x) \subseteq U$.
		 With our definition of $C(x)$, it can be easily seen that $\exists$ has a winning strategy
		 at $(\delta,x)$ in $\game_U$: Firstly, by Lemma~\ref{lem:GvsGU}, for each $U \subseteq X$ we know that
		 $\exists$ has a strategy $f_U$ in $\G_U$
		 at $(\delta,x)$ such that every $\unf_p$-full $\G$-play $\pi$ compliant with $f$ starting at $(\delta,x)$
		 corresponds to a full, $f_U$-compliant $\G_U$-play.
		 % for any clopen set $U \subseteq X$ with
		 %$C(x_0) \subseteq U$: 
		 
		 Suppose now for a contradiction that there is some $U' \in \Clp(\bbX)$ with $C(x) \subseteq U'$
		 for which $(\delta,x) \not\in \Win_\exists(\G_{U'})$. 
		 This implies that the strategy $f_{U'}$ cannot be winning for $\exists$ in $\G_{U'}$ at 
		 $(\delta,x)$ and thus there exists some state
		 $(p,y)$ with $y \not\in U'$ and with the property that $(p,y)$ 
		 is reachable from $(\delta,x)$ in an full $\G$-play $\pi$  
		 compliant
		 $\exists$'s strategy $f_{U'}$. By definition of $f_{U'}$, there
		 exists a $\unf_p$-full $\G$-play $\pi$ from $(\delta,x)$ to $(p,y)$ that is
		 compliant with $f$.
		 This leads to a contradiction: at position $(p,y)$ in $\game$ - as $y \in X \setminus U'$ by assumption -  $\forall$
		 could move to $(p,X \setminus U')$ and $\exists$ could choose an element $z_y \in X \setminus U'$ 
		 and move to $(\delta,z_y)$ according to her strategy $f$.
		 By definition of $C$, we get $z_y \in C(x) \subseteq U'$ and hence $z_y \in U'$ which is a contradiction.
		 
		 This finishes the proof of the fact that $\exists$ has a winning strategy
		 at $(\delta,x)$ in $\game_U$ for any clopen set $U \subseteq X$ with
		 $C(x) \subseteq U$.
		 Consequently, by~(\ref{equ:IH_adequ}), we have
		 $x \in \operator{p}{V}{\delta}(U)$ for all $U \in \Clp(\bbX)$ with $C(x) \subseteq U$.
		 This means that for each $x \in \Delta$, $\exists$ can move from position $x$ to position
		 $C(x)$ in $\game_\mu^I(\operator{x}{V}{\delta})$, i.e., $C$ encodes a
		 legitimate strategy for $\exists$ in all positions $x \in \Delta$.		 
		 We are now going to prove that for any play \vspace{-.2cm}
		 \[ x \; C(x) \; x_1 \; C(x_1) \; x_2 \; C(x_2) \dots x_n \; C(x_n) \vspace{-.2cm}
\]
		 of $\game_\mu^I(\operator{p}{V}{\delta})$ starting in $x$ and compliant with
		 $\exists$ strategy $C$ it is possible to construct a ``shadow'' play of $\game$ starting
		 at $(\varphi,x)$ that is compliant with $\exists$'s winning strategy in $\game$
		 and that is of the form \vspace{-.2cm}
		 \[ (\varphi,x) \dots (\delta,x_1) \dots (\delta,x_2) \dots (\delta,x_n). \vspace{-.2cm}
\] 
		 It suffices to see how a round $x_i \; C(x_i) \; x_{i+1}$ in $\game_\mu^I(\operator{p}{V}{\delta})$
		 is mirrored in $\game$. To this aim note that $x_{i+1} \in C(x_i)$. Hence
		 there exists some $U \in \Clp(\bbX)$ with $x_{i+1} \in U$  such that $(p,U)$ is reachable from $(\delta,x_{i})$ via a $\G$-play
		 $\pi$ compliant with $\exists$'s winning strategy that is 
		 continued by $\eloi$ by moving to position $(\delta,x_{i+1})$.
		 %with $z \in U$ in a
		 %play $\pi$ conform $\exists$'s winning strategy. 
		% and such that $x_{i+1} \in U \subseteq C(x_i)$.
		 Clearly the play $\pi$ followed by $\eloi$'s move to $(\delta,x_{i+1})$ constitutes
		 the required shadow play of $\game$. 
	\end{proof} \vspace{-.1cm}
	
\begin{example}\label{ex:c} {\em 
We will give an example of an extremally disconnected modal space $(\bbX,R)$ with
$\bbX=(X,\tau)$, a clopen valuation $V$ and a modal formula $\varphi(q,p)$ such that
the standard semantics of $\mu q. \varphi$ and the topological semantics of $\mu q.\varphi$ differ.
Let $\mathbb{Z}$ be the set of integers with the discrete topology.
Let $X = \beta(\mathbb{Z})$ be the Stone--\v{C}ech compactification of $\mathbb{Z}$. 
 Then $\beta(\mathbb{Z})$ is extremally disconnected, see eg~\cite{Sik60}.
We define a relation $R$ on $X$ by 
$z R y $ iff ($z,y\in \mathbb{Z}$ and $y = z+1$ or $y=z-1$ or $z\in X$ and $y\in \beta(\mathbb{Z})\setminus \mathbb{Z}$).
%It is easy to check that $(W,R)$ is a descriptive frame. Moreover, by Lemma~\ref{betalem}, $(W,R)$ is 
%a descriptive $\mu$-frame. 
Now we define a clopen valuation $V(p)= \{0\}$. 
Consider the formula $\varphi(q,p)= p\vee \Diamond\Diamond q$. 
The standard semantics of $\mu q.\varphi$ is equal to the set of all even and negative even numbers. 
The topological semantics, in contrast, is equal to the whole space $X$. }
%The closure of this set contains a proper subset of $\beta(\mathbb{Z})\setminus \mathbb{Z}$ and, as is easy to check, is not a fixed point
%of $\varphi(x,h(p))$. 
%Note that in this case we have $\br{\mu x(p\vee \Diamond\Diamond x)}_h^{\Clp(\bbX)}= \br{\mu x(p\vee \Diamond\Diamond x)}_h^{\Cl(W)}= W$.	
%}
\end{example}
\vspace{-.3cm}

%\subsection{Compact Hausdorff spaces: regular open semantics}
%	The game board of the evaluation game $\Ev_r(\bbX,\varphi)$ of some formula $\varphi$
%	on a model $\bbX$ that is based on a modal compact Hausdorff frame is defined very similarly to the evaluation
%	game in the previous section. Rather than writing down the full definition we confine ourselves to making precise where
%	the differences between the games are.

\section{Bisimulations}
We are now going to describe bisimulations for our topological setting.
The definition is essentially the standard one with an additional topological condition.

\begin{definition}
	Let $\bbM_1=(\bbX_1,R_1,V)$ and $\bbM_2=(\bbX_2,R_2,V)$ be extremally disconnected Kripke models based
	on the spaces $\bbX_1 = (X_1,\tau_1)$ and $\bbX_2 = (X_2,\tau_2)$.
	A relation $Z \subseteq X_1 \times X_2$ is called a \emph{clopen bisimulation} iff $Z \subseteq X_1 \times X_2$ is a (standard) Kripke
	bisimulation and for any clopen subsets $U_1 \in \Clp(\bbX_1)$ and $U_2 \in \Clp(\bbX_2)$ we have
	$Z[U_1] = \{ x' \in X_2 \mid \exists x \in U_1. (x,x') \in Z \} \in \Clp(\bbX_2)$ and
	$Z^{-1}[U_2] = \{ x \in X_1 \mid \exists x' \in U_2. (x,x') \in Z \} \in \Clp(\bbX_1)$. 
%	\marginpar{find different notations $Z[U]$ and $Z^-[U]$? Nick: $Z^{-1}$ ?}
\end{definition} 

The justification for the notion of clopen bisimulations is provided by the following proposition.

\begin{proposition}\label{prop:bisim}
	Let $Z$ be a clopen bisimulation between extremally disconnected Kripke models   
	$\bbM_1=(\bbX_1,R_1,V)$ and $\bbM_2=(\bbX_2,R_2,V)$.
%	that are based on extremally disconnected spaces. 
	Then for any formula $\varphi \in \Lang_\mu$ of the modal $\mu$-calculus and any 
	states $x \in  X_1$ and $x' \in X_2$  such that $(x,x') \in Z$, we have $x \in \lsem \varphi \rsem$ iff $x' \in \lsem \varphi \rsem$.
\end{proposition}

\begin{proof}
      Suppose that $(x,x') \in Z$ and that $x \in \lsem \varphi \rsem$ for some formula $\varphi$.
      This implies by our adequacy theorem that $(\varphi,x) \in \Win_\exists(\Ev(\varphi,\bbM_1))$.
      We are now going to transform $\exists$'s winning strategy in $\game_1= \Ev(\varphi,\bbM_1)$ at position
      $(\varphi, x)$ into a winning strategy for $\exists$ in $\game_2= \Ev(\varphi,\bbM_2)$ at position $(\varphi,x')$.
      
      As a preparation we need to define when we consider positions of $\game_1$ and $\game_2$ to be equivalent:
      we say $(\psi_1,x_1) \in \Lang_\mu \times X_1$ and $(\psi_2,x_2) \in \Lang_\mu \times X_1$ are $Z$-equivalent if $\psi_1 = \psi_2$ and $(x_1,x_2) \in Z$.
      Furthermore we write  $(p,U_1) \leq_Z (q,U_2)$ for $(p,U_1) \in \Lang_\mu \times \Clp(\bbX_1)$ and $(p,U_2) \in \Lang_\mu \times \Clp(\bbX_2)$ if $p = q$ and if
      for all $x \in U_1$ there exists $x' \in U_2$ such that $(x,x') \in Z$. Similarly we define $(p,U_1) \geq_Z (q,U_2)$.
       Consider two (possibly partial) plays $\pi_1= b_1 \ldots b_k$ and $\pi_2=b_1'\dots b_l'$ of $\game_1$ and $\game_2$, respectively.
      We say $\pi_1$ and $\pi_2$ are $Z$-equivalent iff $k=l$ and for all $i=1,\ldots,k$ we have 
      \begin{itemize}
      	\item $b_i$ and $b_i'$ are of the form $b_i=(\psi,x)$ and $b_i'=(\psi,x')$ and both positions are $Z$-equivalent, or
	\item $b_i = (p,U_1)$, $b_i' = (p,U_2)$, $p$ is bound by $\mu$ and $(p,U_1) \leq_Z (p,U_2)$, or
	\item $b_i = (p,U_2)$, $b_i' = (p,U_2)$, $p$ is bound by $\nu$ and $(p,U_1) \geq_Z (p,U_2)$. 
	\end{itemize}
      
      Let $\pi_1$ be a play of $\game_1$ that starts in $\eloi$'s winning position $(\varphi,x)$ and that is played according to $\eloi$'s
      winning strategy. We are going to show that if $\pi_2$ is a $Z$-equivalent play of $\game_2$
      that starts at position $(\varphi,x')$, then either
      \begin{itemize}
      	\item  both plays $\pi_1$ and $\pi_2$ are full (and thus won by $\eloi$) or 
	\item it is $\eloi$'s turn and $\eloi$ has a strategy to extend $\pi_2$ to a play $\pi_2 b'$ that is $Z$-equivalent to an extension $\pi_1 b$ of $\pi_1$ such that $\pi_1 b$ is a $\game_1$-play compliant with $\exists$'s winning strategy, or
	\item it is $\abe$'s turn and for all of $\abe$'s moves that extend $\pi_2$ to $\pi_2 b'$ there is a move of $\abe$ in $\game_1$ such that
		 the resulting play $\pi_1 b$ of $\game_1$ is $Z$-equivalent to $\pi_2 b'$.
     \end{itemize} 
       Clearly this claim will imply that $\eloi$ has a winning strategy in $\game_2$ at position $(\varphi,x')$ as required.
%      \end{itemize}
      The claim is proven by a case distinction on the last state of $\pi_2$. 
      Due to space reasons we only discuss the cases of the modal diamond and the (least) fixpoint cases.
      
      	\noindent
	{\bf Case:} $\pi_2 = b_1'\dots b_n' (\Diamond \psi,x_2)$. By assumption there exists a $Z$-equivalent play $\pi_1 = b_1\dots b_n (\Diamond \psi,x_1)$ which in particular implies
           that $(x_1,x_2) \in Z$. Clearly it is $\eloi$'s turn and she can prolong the $\game_1$-play by moving according to her strategy to $(\psi,y)$ for some $y \in X_1$ with $(x_1,y) \in R_1$.
           As $Z$ is a bisimulation we know that there must be $y' \in X_2$ such that $(x_2,y') \in R_2$ and $(y.y') \in Z$. Hence $\abe$ can prolong the $\pi_2$-play by moving
           to $(\psi_1,y')$ and the resulting plays $\pi_1 = b_1\dots b_n (\Diamond \psi,x_2)(\psi,y)$ and $\pi_2 = b_1'\dots b_n' (\Diamond \psi,x_1) (\psi,y')$ are $Z$-equivalent.
      	  
	   \noindent
	  {\bf Case:} $\pi_2 = b_1' \dots b_n' (p,x_2)$ for some $p \in \BV(\varphi)$ that is bound by a $\mu$-operator.
	   In this case $\pi_1 = b_1 \dots b_n (p,x_1)$ and its $\abe$'s turn to continue both plays. Let $\abe$'s move in $\game_2$ be to $(p,U_2)$ for some
	   clopen subset $U \in \Clp(\bbX_2)$ with $x_2 \in U_2$ . Because $(x_1,x_2) \in Z$ and by the definition of a clopen bisimulation
	   we have that $U_1 \mathrel{:= }Z^{-1}[U_2]$ is a clopen neighbourhood of $x_1$. Therefore $\abe$ could extend the $\game_1$-play by moving
	   to $(p,U_1)$. The resulting plays  $\pi_1 = b_1 \dots b_n (p,x_1) (p,U_1)$ and $\pi_2 = b_1' \dots b_n' (p,x_2)(p,U_2)$ are clearly $Z$-equivalent
	   because all elements $U_1$ have their $Z$-correspondant in $U_2$ and hence we have $(p,U_1) \leq_P (p,U_2)$ as required. 
	   
	  \noindent
	  {\bf Case:} $\pi_2 = b_1' \dots b_n' (p,U_2)$ for some $p \in \BV(\varphi)$ that is bound by a $\mu$-operator.
	  	By assumption we have a $Z$-equivalent $\game_1$-play  $\pi_1 = b_1 \dots b_n (p,U_1)$ with the property
		that $(p,U_1) \leq_Z (p,U_2)$. By the definition of the game it is clear that in both plays $\eloi$ has to move.
		She continues $\pi_1$ by moving according to her winning strategy to some $(\psi,y)$ with $y \in U_1$. By definition of $\leq_Z$ there
		exists a $y' \in U_2$ such that $(y,y') \in Z$ and hence $\eloi$ can extend the play $\pi_2$ by moving to $(\psi,y')$. Again
		the resulting plays $\pi_1 =  b_1 \dots b_n (p,U_1)(\psi,y)$ and $\pi_2 =  b_1 \dots b_n (p,U_2)(\psi,y')$ are obviously $Z$-equivalent. 
 %     \end{description}
      The other cases of the induction can be dealt with in a similar fashion.
%      are discussed in the appendix.
%      \begin{description}
 This shows that from $x \in \lsem \varphi \rsem$ and $(x,x') \in Z$ we are able to deduce  $x' \in \lsem \varphi \rsem$.
      The implication in the opposite direction can be proven in a completely symmetrical way. As $\varphi$ was arbitrary we conclude that clopen bisimilarity implies
      equivalence with respect to the topological modal $\mu$-calculus. %\vspace{-.4cm}
\end{proof}

\begin{rema}
    {\em We leave it open whether the converse of Proposition~\ref{prop:bisim}  also holds, i.e.,  whether we have a Hennessy-Milner property wrt our notion of clopen bisimulation. %We plan to investigate this in the near future. %but conjecture that the Hennessy-Milner property fails.
    A closely related question is how our clopen bisimulations compare to the Vietoris bisimulations of~\cite{befove10:viet}. It is obvious that
    the topological closure of a clopen bisimulation is a Vietoris bisimulation and hence that clopen bisimilarity implies
    Vietoris bisimilarity. Proving the converse would yield the Hennessy-Milner property with regard to clopen bisimilarity
    as a corollary of~\cite[Cor.~3.10]{befove10:viet}.}
\end{rema}

%%\begin{rema}
%    At this stage it is not clear whether or not the converse of the Proposition~\ref{prop:bisim} 
%    %-   - 
%    also
%    holds true, ie.  whether or not a Hennessy-Milner property wrt to our notion of bisimulation holds. We plan to investigate this in the near future. %but conjecture that the Hennessy-Milner property fails.
%    A closely related question is how our clopen bisimulations compare to the Vietoris bisimulations from~\cite{befove10:viet}. It is obvious that
%    the topological closure of a clopen bisimulation is a Vietoris bisimulation and hence that clopen bisimilarity implies
%    Vietoris bisimilarity. Proving the converse would yield the Hennessy-Milner property with regard to clopen bisimilarity
%    as a corollary of~\cite[Cor.~3.10]{befove10:viet}.
%\end{rema}
\section{Conclusions and future work}
%We conclude:
%
%\begin{itemize}
%	\item first step towards  continuous game semantics of $\mu$-calculus (?)
%	\item phrase the whole story more abstractly (closure operators?)
%	\item decidability? (maybe do not allow all clopens/regular opens in the game
%	      in order to obtain decidable approximation of semantics?)
%	\item automata
%	\item completeness?
%\end{itemize}

In this paper we developed game semantics for topological fixpoint logic on extremally disconnected modal spaces.
These results can be seen as first steps towards the theory of topological fixpoint logic in general
and towards admissible game semantics of $\mu$-calculus in particular. 
As next steps we intend to extend
this framework to other classes of descriptive $\mu$-frames and to devise automata that operate on Kripke frames over topological
spaces. This will provide a deeper understanding of these structures as well as of axiomatic
systems of the modal $\mu$-calculus, since axiomatic systems of the $\mu$-calculus are complete wrt 
descriptive $\mu$-frames. Other important questions concern the finite model property, decidability and computational complexity and other key properties of topological fixpoint logics.

A further interesting research direction is to investigate modal fixpoint logic of Kripke frames based on compact
Hausdorff spaces and beyond. However, instead of clopen sets we will have to work with regular open
sets in this setting. This means we will enter the realm of modal compact Hausdorff spaces introduced in~\cite{BBH15}. These
are exactly the spaces that correspond to coalgebras for the Vietoris functor on the category of compact
Hausdorff spaces. Sahlqvist fixpoint correspondence for such spaces has been developed
already in~\cite{BS15}. I
This approach could pave
the way for an expressive and  decidable fixpoint logic for the verification of continuous systems or, more
generally, systems that combine discrete and continuous systems such as hybrid automata~\cite{DBLP:conf/lics/Henzinger96}.

Finally, we want to clarify the connection of our work to topological games \`{a} la 
Banach-Mazur~\cite{telga:topo87}. These games are similar to our fixpoint games as players move by choosing e.g.~open subsets - the fundamental differences are i) they characterise properties of the topology rather than properties of a relational structure over a topological space and ii) our parity winning condition that ensures determinacy. 

\vspace{-0.2cm}

\bibliographystyle{eptcs}
\bibliography{Mybib,automata}

%\newpage
%\appendix
%\input{appendix}

\end{document}